\documentclass{article} 
\usepackage{nips13submit_e,times}
\usepackage{hyperref}
\usepackage{url}

\usepackage{amsmath,amssymb,amsthm}
\usepackage{epsfig,graphics}
\usepackage{multicol,subfigure}
\usepackage{wrapfig, rotating}

\usepackage{algorithmic,algorithm}
\usepackage{amsthm}

\newcommand{\barr}{\left[ \begin{array} }
\newcommand{\earr}{ \end{array} \right] }

\newcommand{\ars}[1]{\left[ \begin{array}{#1}}
\newcommand{\are}{\end{array} \right] }
\newcommand{\oars}[1]{\begin{array}{#1}}
\newcommand{\oare}{\end{array}}

\newcommand{\eqs}{\begin{eqnarray}}
\newcommand{\eqe}{\end{eqnarray}}
\newcommand{\eqsn}{\begin{eqnarray*}}
\newcommand{\eqen}{\end{eqnarray*}}

\newcommand{\ens}{\begin{enumerate}}
\newcommand{\ene}{\end{enumerate}}

\newcommand{\its}{\begin{itemize}}
\newcommand{\ite}{\end{itemize}}

\newcommand{\des}{\begin{description}}
\newcommand{\dee}{\end{description}}

\usepackage[mathscr]{eucal}

\usepackage{amsbsy}

\usepackage{bm}

\usepackage{paralist}

\usepackage{xspace}

\newcommand{\Sec}[1]{\hyperref[sec:#1]{\S\ref*{sec:#1}}} 
\newcommand{\Eqn}[1]{\hyperref[eq:#1]{(\ref*{eq:#1})}} 
\newcommand{\Fig}[1]{\hyperref[fig:#1]{Figure~\ref*{fig:#1}}} 
\newcommand{\Tab}[1]{\hyperref[tab:#1]{Table~\ref*{tab:#1}}} 
\newcommand{\Thm}[1]{\hyperref[thm:#1]{Theorem~\ref*{thm:#1}}} 
\newcommand{\Lem}[1]{\hyperref[lem:#1]{Lemma~\ref*{lem:#1}}} 
\newcommand{\Prop}[1]{\hyperref[prop:#1]{Property~\ref*{prop:#1}}} 
\newcommand{\Cor}[1]{\hyperref[cor:#1]{Corollary~\ref*{cor:#1}}} 
\newcommand{\Def}[1]{\hyperref[def:#1]{Definition~\ref*{def:#1}}} 
\newcommand{\Alg}[1]{\hyperref[alg:#1]{Algorithm~\ref*{alg:#1}}} 
\newcommand{\Ex}[1]{\hyperref[ex:#1]{Example~\ref*{ex:#1}}} 

\newcommand{\Real}{\mathbb{R}}

\newcommand{\Tra}{^{\rm T}} 

\newcommand{\M}[1]{{\bm{\mathbf{\MakeUppercase{#1}}}}} 

\newcommand{\T}[1]{\boldsymbol{\mathscr{\MakeUppercase{#1}}}} 

\newcommand{\TV}{\T{V}}

\newtheorem{theorem}{Theorem}[section]
\numberwithin{theorem}{subsection}
\newtheorem{defn}{\indent \bf Definition}[section]
\numberwithin{defn}{subsection}
\newtheorem{conjecture}{Conjecture}[section]
\numberwithin{conjecture}{subsection}
\newtheorem{remark}{Remark}[section]
\numberwithin{remark}{subsection}

\title{Novel Factorization Strategies for Higher Order Tensors: Implications for Compression and Recovery of Multi-linear Data  }

\author{
Zemin Zhang, Gregory Ely, Shuchin Aeron \\
Department of ECE\\
Tufts University	\\
Medford, MA 2155 \\
\texttt{jamie.zeminzhang@gmail.com}\\
\texttt{gregoryely@gmail.com}  \\
\texttt{shuchin@ece.tufts.edu} \\
\And
Ning Hao and Misha Kilmer \\
Department of Mathematics \\
Tufts University \\
Medford, MA 02155\\
\texttt{ning.hao@tufts.edu}     \\
\texttt{misha.kilmer@tufts.edu} \\
}

\nipsfinalcopy 

\begin{document}

\maketitle

\begin{abstract}
In this paper we propose novel methods for compression and recovery of multilinear data under limited sampling. We exploit the recently proposed tensor-Singular Value Decomposition (t-SVD)\cite{KilmerBramanHao2011}, which is a group theoretic framework for tensor decomposition.  In contrast to popular existing tensor decomposition techniques such as higher-order SVD (HOSVD), t-SVD has optimality properties similar to the truncated SVD for matrices. Based on t-SVD, we first construct novel tensor-rank like measures to characterize informational and structural complexity of multilinear data. Following that we outline a complexity penalized algorithm for tensor completion from missing entries. As an application, 3-D and 4-D (color) video data compression and recovery are considered. We show that videos with linear camera motion can be represented more efficiently using t-SVD compared to traditional approaches based on vectorizing or flattening of the tensors. Application of the proposed tensor completion algorithm for video recovery from missing entries is shown to yield a superior performance over existing methods. In conclusion we point out several research directions and implications to online prediction of multilinear data.

\end{abstract}

\vspace{-5mm}
\section{Introduction}
\vspace{-2mm} 
This paper focuses on several novel methods for robust recovery of multilinear signals or tensors (essentially viewed as 2-D, 3-D,..., N-D data) under limited sampling and measurements. Signal recovery from partial measurements, sometimes also referred to as the problem of data completion for specific choice of measurement operator being a simple downsampling operation, has been an important area of research, not only for statistical signal processing problems related to inversion, \cite{Ely_ICASSP2013, Wright12-II,Ji_PAMI12}, but also in machine learning for online prediction of ratings, \cite{Hazan_JMLR12}. All of these applications exploit low structural and informational complexity of the data, expressed either as low rank for the 2-D matrices \cite{Ely_ICASSP2013,Wright12-II}, which can be extended to higher order data via \emph{flattenning or vectorizing} of the tensor data such as tensor N-rank \cite{GandyRY2011}, or other more general tensor-rank measures based on particular tensor decompositions such as higher oder SVD (HOSVD) or Tucker-3 and Canonical Decomposition (CANDECOMP). See \cite{SIREV} for a survey of these decompositions.

The key idea behind these methods is that under the assumption of low-rank of the underlying data thereby constraining the complexity of the hypothesis space, it should be feasible to recover data (or equivalently predict the missing entries) from number of measurements in proportion to the rank. Such analysis and the corresponding identifiability results are obtained by considering an appropriate complexity penalized recovery algorithm under observation constraints, where the measure of complexity, related to the notion of rank, comes from a particular \emph{factorization} of the data. Such algorithms are inherently combinatorial and to alleviate this difficulty one looks for the tightest convex relaxations of the complexity measure, following which the well developed machinery of convex optimization as well as convex analysis can be employed to study the related problem. For example, rank of the 2-D matrix being relaxed to the Schatten 1-norm, \cite{Watson92} and tensor $N$-rank for order $N > 2$ tensors being relaxed to overlapped Schatten p-norms, \cite{GandyRY2011}. Similarly one can extend this approach to HOSVD and take the sum of absolute values of the entries of the core tensor as a (relaxed) measure of data complexity.

Note that all of the current approaches to handle multilinear data extend the nearly optimal 2-D SVD\footnote{Optimality of 2-D SVD is based on the optimality of truncated SVD as the best $k$-dimensional $\ell_2$ approximation.} based vector space approach to the higher order ($N > 2$) case. This results in loss of optimality in the representation.  In contrast, our approach is based upon recent results on decomposition/factorization of tensors in \cite{Braman2010,KilmerMartin2010,KilmerBramanHao2011} in which the authors refer to as tensor-SVD or t-SVD for short. In the t-SVD framework we view 3D tensors as a matrix of tubes (in the third dimension) and define a commutative operation (convolution) between the tubes. This commutative structure leads to viewing the tensor tensor multiplication as a simple matrix-matrix multiplication where the multiplication operation is defined via the commutative operation. With this construction, one can now introduce the notion of a t-SVD which is similar to the traditional SVD, see Figure~\ref{fig:tSVD}. In this context we identify two (relaxed) measures of structural complexity, namely the Tensor-Nuclear-Norm (TNN) first used in \cite{semerci1} for multi-energy CT application and Tensor Tubal Norm (TTN) which extends the notion of SVD based rank to the number of non-zero singular tubes in t-SVD. 

This paper is organized as follows. Section \ref{sec:tSVD} presents the notations and provide an overview and key results on t-SVD from \cite{Braman2010, KilmerMartin2010, KilmerBramanHao2011} and illustrates the key differences and advantages over other tensor decomposition methods. We will propose measures of structural complexity based on t-SVD. Using these measures, in Section \ref{sec:tcompletion} we introduce a complexity penalized algorithm for tensor completion from randomly sampled data cube. As an application we consider video data completion from missing entries and Section \ref{sec:video_app} shows the compression performance of t-SVD based representation on several video data sets and evaluates the performance of the proposed algorithm for video completion then compares it to current approaches. Finally we outline implications of this work and future research directions in Section \ref{sec:future_work}.

\vspace{-2mm} 

\section{Brief overview of t-SVD}
\label{sec:tSVD}
\vspace{-2mm} 
In this section, we describe the tensor decomposition as proposed in \cite{Braman2010,KilmerMartin2010,KilmerBramanHao2011} and the notations used throughout the paper. 
\vspace{-2mm} 
\subsection{Notation and Indexing}
\vspace{-2mm} 
A \emph{\textbf{Slice}} of an N-dimensional tensor is a 2-D section defined by fixing all but two indices. A \emph{\textbf{Fiber}} of an N-dimensional tensor is a 1-D section defined by fixing all indices but one \cite{SIREV}. For a third order tensor $\T{A}$, we will use the Matlab notation $\T{A}(k,:,:)$ , $\T{A}(:,k,:)$ and $\T{A}(:,:,k)$ to denote the $k_{th}$ horizontal, lateral and frontal slices, and $\T{A}(:,i,j)$, $\T{A}(i,:,j)$ and $\T{A}(i,j,:)$ to denote the $(i,j)_{th}$ mode-1, mode-2 and mode-3 fiber.  In particular, we use $\T{A}^{(k)}$ to represent $\T{A}(:,:,k)$.
\begin{defn}
\emph{\textbf{Tensor Transpose}}.  Let $\T{A}$ be a tensor of size $n_1 \times n_2 \times n_3$, then $\T{A} \Tra$ is the $n_2 \times n_1 \times n_3$ tensor obtained by transposing each of the frontal slices and then reversing the order of transposed frontal slices $2$ through $n_3$.
\end{defn}
\begin{defn} \emph{\textbf{Identity Tensor}}. The identity tensor $\T{I} \in \mathbb{R}^{n_1 \times n_1 \times n_3}$ is a tensor whose first frontal slice is the $n_1 \times n_1$ identity matrix and  all other frontal slices are zero.
\end{defn}
\begin{defn} \emph{\textbf{f-diagonal Tensor}}. A tensor is called f-diagonal if each frontal slice of the tensor is a diagonal matrix.
\end{defn}
We can view a 3-D tensor of size $n_1 \times n_2 \times n_3$ as an $n_1 \times n_2$ matrix of tubes. By introducing a commutative operation $*$ between the tubes $\mathbf{a}, \mathbf{b}\in \Real^{1 \times 1 \times n_3}$  via $\mathbf{a} * \mathbf{b} = \mathbf{a} \circ \mathbf{b}$ where $\circ$ denotes the \emph{circular convolution} between the two vectors, we can define the t-product between two tensors.  This multiplication is analogous to the matrix multiplication except that circular convolution replaces the multiplication operation between the elements (which are tubes now).
\begin{defn} \emph{\textbf{t-product}}.  
The t-product $\T{C}$ of $\T{A} \in \mathbb{R}^{n_1 \times n_2 \times n_3}$ and $\T{B} \in \mathbb{R}^{n_2\times n_4 \times n_3}$ is a tensor of size $n_1 \times n_4 \times n_3$ where the $(i,j)_{th}$ tube denoted by $\T{C}(i,j,:)$ for $i = 1,2, ..., n_1$ and $j = 1, 2,..., n_4$ of the tensor $\T{C}$ is given by $\sum_{k=1}^{n_2} \T{A}(i,k,:) * \T{B}(k,j,:)$.
\end{defn}
The t-product of $\T{A}$ and $\T{B}$ can be computed efficiently by performing the fast Fourier transformation (FFT) along the tube fibers of $\T{A}$ and $\T{B}$ to get $\hat{\T{A}}$ and $\hat{\T{B}}$,  multiplying the each pair of the frontal slices of $\hat{\T{A}}$ and $\hat{\T{B}}$ to obtain $\hat{\T{C}}$, and then taking the inverse FFT along the third mode to get the result. For details about the computation, see \cite{KilmerMartin2010, KilmerBramanHooverHao, HaoKilmerBramanHoover2012}.
\begin{defn} \emph{\textbf{Orthogonal Tensor}}. A tensor $\T{Q}\in\mathbb{R}^{n_1\times n_1\times n_3}$ is orthogonal if
\begin{equation}
\T{Q} \Tra * \T{Q} = \T{Q} * \T{Q} \Tra  = \T{I}
\end{equation}
\end{defn}

\subsection{Tensor Singular Value Decomposition (t-SVD)}
\vspace{-2mm} 
The new t-product  allows us to define a tensor Singular Value Decomposition (t-SVD). 
\begin{theorem}
For $\T{M}\in\mathbb{R}^{n_1 \times n_2 \times n_3}$, the t-SVD of $\T{M}$ is  given by
\begin{equation}
\T{M} = \T{U} *\T {S} *\T{V}\Tra
\end{equation}
\noindent where $\T{U}$ and $\TV$ are orthogonal tensors of size $n_1 \times n_1 \times n_3 $ and  $n_2 \times n_2 \times n_3 $ respectively. $\T{S}$ is a rectangular $f$-diagonal tensor of size $n_1 \times n_2 \times n_3 $, and $*$ denotes the t-product. 
\end{theorem}
\vspace{-1.5mm}
We can obtain this decomposition by computing matrix SVDs in the Fourier domain, see Algorithm 1. The notation in the algorithm can be found in\cite{Martin_SIAM2013}. 
Figure \ref{fig:tSVD} illustrates the decomposition for the 3-D case.
\begin{algorithm}
\label{alg:tSVD - for $N$-D tensor}
  \caption{t-SVD}
  \begin{algorithmic}
  \STATE \textbf{Input: } $\T{M} \in \mathbb{R}^{n_1 \times n_2  ... \times n_N}$
  \STATE $\rho = n_3n_4...n_N$
   \FOR{$i = 3 \hspace{2mm} \rm{to} \hspace{2mm} N$}
  	\STATE ${\T{D}} \leftarrow \rm{fft}(\T{M},[\hspace{1mm}],i)$;
  \ENDFOR
  
  \FOR{$i = 1 \hspace{2mm} \rm{to} \hspace{2mm} \rho$}
  	\STATE $ [\M{U}, \M{S}, \M{V}] = svd(\T{D}(:,:,i))$
  	\STATE $ {\hat{\T{U}}}(:,:,i) = \M{U}; \hspace{1mm} {\hat{\T{S}}}(:,:,i) = \M{S}; \hspace{1mm} {\hat{\T{V}}}(:,:,i) = \M{V}; $
  \ENDFOR
  
  \FOR{$i = 3 \hspace{2mm} \rm{to} \hspace{2mm} \rho$}
  	\STATE $\T{U} \leftarrow \rm{ifft}(\hat{\T{U}},[\hspace{1mm}],i); \hspace{1mm} \T{S} \leftarrow \rm{ifft}(\hat{\T{S}},[\hspace{1mm}],i); \hspace{1mm} \T{V} \leftarrow \rm{ifft}(\hat{\T{V}},[\hspace{1mm}],)i$;
  \ENDFOR
  \end{algorithmic}
\end{algorithm}

%
%

\begin{figure}[htbp]
\centering \makebox[0in]{
    \begin{tabular}{c c}
      \includegraphics[scale=0.55]{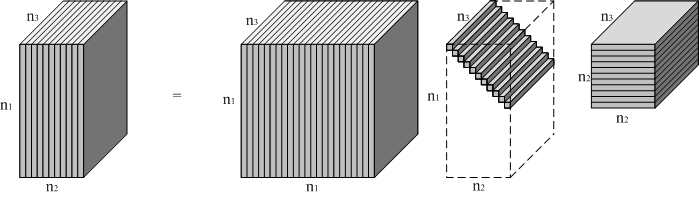}
 \end{tabular}}
  \caption{ The t-SVD of an $n_1 \times n_2 \times n_3$ tensor.}
  \label{fig:tSVD}
\end{figure}

\vspace{-2mm} 
\subsection{t-SVD: Fundamental theorems and key results}
\vspace{-2mm} 
The two widely used tensor decompositions, Tucker and PARAFAC\cite{SIREV} are usually seen as a higher order SVD for tensors. Both of these decompositions have several disadvantages. In particular, one cannot  easily determine the rank-one components of the PARAFAC decomposition and given a fixed rank, calculation of an approximation can be  numerically unstable. The tensor-train form of the Tucker decomposition is studied in \cite{Oseledets} as an alternative form. Tucker decomposition can be seen as a generalization of PARAFAC decomposition, and the truncated decomposition doesn't yield the best fit of the original tensor. In contrast, the t-SVD can be easily computed by solving several SVDs in the Fourier domain. More importantly, it gives an optimal approximation of a tensor measured by the Frobenius norm of the difference, as state in the following theorem\cite{KilmerMartin2010, KilmerBramanHooverHao, HaoKilmerBramanHoover2012}. 

\begin{theorem} Let the t-SVD of $\T{M} \in \mathbb{R}^{n_1 \times n_2 \times n_3}$ be 
given by $\T{M} = \T{U} * \T{S} * \T{V}^T$ and for $k < \min(n_1,n_2)$ define
$\T{M}_k = \sum_{i=1}^{k} \T{U}(:,i,:) * \T{S}(i,i,:) * \T{V}(:,i,:)^T $, 
Then
$$\T{M}_k = \arg \min_{\tilde{\T{M}} \in \mathbb{M}} \| \T{M} - \tilde{\T{M}} \|_F $$
where $\mathbb{M} = \{ \T{C} = \T{X} * \T{Y} |
\T{X} \in \mathbb{R}^{n_1 \times k \times n_3}, \T{Y} \in \mathbb{R}^{k 
\times n_2 \times n_3} \}$.
\end{theorem}
\vspace{-2mm} 
\subsection{Measures of tensor complexity using t-SVD }
\vspace{-2mm} 
We  now define two measures of tensor complexity based on the proposed t-SVD: the tensor multi-rank, proposed in  \cite{KilmerBramanHooverHao}, and the novel tensor tubal rank.
\begin{defn} \textbf{Tensor multi-rank}. The multi-rank of $\T{A}\in\mathbb{R}^{n_1\times n_2\times n_3}$ is a vector $p$ $\in \Real^{n_3 \times 1}$ with the $i_{th}$ element equal to the rank of the $i_{th}$ frontal slice of $\hat{\T{A}}$ obtained by taking the Fourier transform along the third dimension of the tensor. 
\end{defn}
One can obtain a scalar measure of complexity as the $\ell_1$ norm of the tensor multi-rank. We now define another measure motivated by the matrix SVD. 
\begin{defn} \textbf{Tensor tubal-rank}. The tensor tubal rank of a 3-D tensor is defined to be the number of non-zero tubes of $\T{S}$ in the t-SVD factorization.
\end{defn}
As in the matrix case,  practical applications of these complexity measures require adequate convex relaxations. To this end we have the following result for the Tensor multi-rank. 

\begin{theorem}
The tensor-nuclear-norm (TNN) denoted by $||\T{A}||_{TNN}$ and defined as the sum of the singular values of all the frontal slices of $\hat{\T{A}}$ is a norm and is the tightest convex relaxation to $\ell_1$ norm of the tensor multi-rank.
\end{theorem}
\begin{proof}
\vspace{-3mm}
The proof that TNN is a valid norm can be found in \cite{semerci1}. The $\ell_1$ norm of the tensor multi-rank is equal to $\text{rank}(\text{blkdiag}(\hat{\T{A}}))$,
for which the tightest convex relaxation is the the nuclear norm of  $\text{blkdiag}(\hat{\T{A}})$ which is TNN of $\T{A}$ by definition. Here $\text{blkdiag}(\hat{\T{A}})$ is a block diagonal matrix defined as follows:
\begin{equation}
\label{eq:blkdiag}
\text{blkdiag}( \hat{\T{A}} ) = 
  \left[\begin{array}{cccc}\hat{{\T{A}}}^{(1)}& & & \\
 & \hat{{\T{A}}}^{(2)} & & \\
 & &\ddots & \\
  & & & \hat{{\T{A}}}^{(n_3)}\end{array} \right]
\end{equation}
where $\hat{\T{A}}^{(i)}$ is the $i_{th}$ frontal slice of $\hat{\T{A}}$, $i = 1,2,...,n_3$.
\end{proof}
 Similar to the TNN we can seek a relaxation for the tensor-tubal-rank, we call tensor-tubal-norm (TTN). It is defined as the sum of $\ell_2$ norms of the tubes $\T{S}(i,i,:)$. However at this point we haven't proven that the TNN is a norm. Based on simulations and tests on a variety of cases we have the following conjecture.  
\begin{conjecture}
The tensor-tubal-norm is a norm and is the tightest convex relaxation to tensor-tubal-rank.
\end{conjecture}
\vspace{-1.5mm}
We conclude this section with the following important remark. 
\begin{remark}
It may appear that our approach can be generalized by taking any orthogonal transformations along third dimension followed by sampling each slice in the transformed domain. However this perspective cannot be easily extended in the t-SVD  framework, as the notions of transpose does not extend to all transforms. In fact, there is an underlying algebraic framework, namely that of viewing in the third order tensors as linear operators on a \emph{module} \cite{MacLaneBook} of oriented matrices with multiplication defined over a commutative ring with identity \cite{Braman2010}. For sake of brevity we will not discuss these aspects in this paper. 
\end{remark}

\vspace{-2mm} 
\vspace{-3mm} 
\section{Tensor completion/prediction from missing observations}
\vspace{-2mm} 
\label{sec:tcompletion}
We will show the case when the tensor data is simply decimated randomly or down sampled in this section. Specifically we consider the problem of data completion from missing entries for multilinear signals.  Suppose there is an unknown tensor $\T{M}$ of size $n_1 \times n_2 \times n_3$ which is {\bf \emph{assumed to have a low tubal-rank}} and we are given a subset of entries $ \{ \T{M}_{ijk}:(i,j,k) \in \bold{\Omega} \}$ where $\bold{\Omega}$ is an indicator tensor of size $n_1 \times n_2 \times n_3$. Our objective is to recover the entire $\T{M}$. This section develops an algorithm for addressing this problem via solving the following complexity penalized algorithm:
\begin{equation}
\begin{aligned}
\label{eq:originalTNN}
\mbox{min} \hspace{4mm}&\|\T{X}\|_{TNN} \\
\mbox{subject to } \,\,\,&P_{\bold{\Omega}} ( \T{X} ) = P_{\bold{\Omega}} ( \T{M} )
\end{aligned}
\end{equation}
\noindent where $P_\bold{\Omega}$ is the orthogonal projector onto the span of tensors vanishing outside of $\bold{\Omega}$. So the $(i,j,k)_{th}$ component of $ P_{\bold{\Omega}} ( \T{X} )$ is equal to $\T{M}_{ijk}$ if $(i,j,k) \in \bold{\Omega}$ and zero otherwise. Let $\T{Y}$ be the available (sampled) data: $\T{Y} = P_\bold{\Omega} \T{M}$. Define $\mathcal{G}=\mathscr{F}_3 P_\bold{\Omega} \mathscr{F}_{3}^{-1}$ where $\mathscr{F}_3$ and $\mathscr{F}^{-1}_3$ are the operators representing the Fourier and inverse Fourier transform along the third dimension of tensors. Then we have
$\hat{\T{Y}} = \mathcal{G} (\hat{\T{M}})$ where $\hat{\T{Y}}$ and $\hat{\T{M}}$ are the Fourier transforms of $\T{Y}$ and $\T{M}$ along the third mode.
So (\ref{eq:originalTNN}) is equivalent with the following:
\begin{equation}
\begin{aligned}
\label{eq:fourierdomain_min}
\mbox{min} \hspace{4mm}&||\mbox{blkdiag}(\hat{\T{X}}) ||_*\\
\mbox{subject to } \,\,\,\,\,&\hat{\T{Y}} = {\cal G}(\hat{\T{X}}) 
\end{aligned}
\end{equation}
where $\hat{\T{X}}$ is the Fourier transform of $\T{X}$ along the third dimension and $\text{blkdiag}(\hat{\T{X}})$ is defined in (\ref{eq:blkdiag}). Noting that $\|\T{X}\|_{TNN} = ||\text{blkdiag}(\hat{\T{X}}) ||_*$. To solve the optimization problem, one can re-write (\ref{eq:fourierdomain_min}) equivalently as follows: 
\begin{equation}
\begin{aligned}
\min \hspace{4mm}& ||\text{blkdiag}( \hat{\T{Z}} )||_* + \mathbf{1}_{\hat{\T{Y}} = {\cal G}(\hat{\T{X}})} \\
\mbox{subject to } \,\,\,\,\, &\hat{\T{X}} - \hat{\T{Z}} = 0  
\end{aligned}
\end{equation}
where $\mathbf{1}$ denotes the indicator function. Then using the general framework of Alternating Direction Method of Multipliers(ADMM)\cite{Boyd_ADMM} we have the following recursion, 
\begin{align}
\label{eq:proj}
\hat{\T{X}}^{k+1}& = \arg\min_{\hat{\T{X}}} \left\{ \mathbf{1}_{\T{Y} = {\cal G}(\hat{\T{X}})} + \langle \hat{\T{Q}}^{k}, \hat{\T{X}} \rangle  + \frac{1}{2} || \hat{\T{X}} - \hat{\T{Z}}^{k}||_{2}^{2} \right\}\nonumber \\
& = \arg\min_{\hat{\T{X}}: \T{Y} = {\cal G}(\hat{\T{X}})} \left\{|| \hat{\T{X}} - (\hat{\T{Z}}^{k} - \hat{\T{Q}}^{k})||_{2}^{2} \right\}\\
\label{eq:shrink}
\hat{\T{Z}}^{k+1}    & = \arg \min_{\T{Z}} \left\{ \frac{1}{\rho} ||\text{blkdiag}(\hat{\T{Z}})||_* + \frac{1}{2} ||  \hat{\T{Z}} - (\hat{\T{X}}^{k+1} + \hat{\T{Q}}^{k})||_{2}^{2} \right\}\\
\hat{\T{Q}}^{k+1} &= \hat{\T{Q}}^{k} + \left(\hat{\T{X}}^{k+1} - \hat{\T{Z}}^{k+1}\right)
\end{align}
where Equation~(\ref{eq:proj}) is least-squares projection onto the constraint and the solution to Equation~(\ref{eq:shrink}) is given by the singular value thresholding\cite{Watson92,Cai_SVT}. 

\textbf{Equivalence of the algorithm to \emph{iterative singular-tubal shrinkage via convolution}}:
According to the particular format that (\ref{eq:shrink}) has, we can break it up into $n_3$ independent minimization problems. Let $\hat{\T{Z}}^{k+1,(i)}$ denotes the $i_{th}$ frontal slice of $\hat{\T{Z}}^{k+1}$. Similarly define $\hat{\T{X}}^{k+1,(i)}$ and $\hat{\T{Q}}^{k,(i)}$. Then (\ref{eq:shrink}) can be separated as:
\begin{equation}
\begin{aligned}
\label{eq:shrink2}
\hat{\T{Z}}^{k+1,(i)}   & = \arg \min_{W} \left\{ \frac{1}{\rho} ||W||_* + \frac{1}{2} ||  W - (\hat{\T{X}}^{k+1,(i)} + \hat{\T{Q}}^{k,(i)})||_{2}^{2} \right\}
\end{aligned}
\end{equation}
for $i = 1,2,...,n_3$. This means each $i_{th}$ frontal slice of $\hat{\T{Z}}^{k+1}$ can be calculated through (\ref{eq:shrink2}). The solution to (\ref{eq:shrink2}) can be found in\cite{Watson92,Cai_SVT}. In detail, if $U S V^\text{T} = (\hat{\T{X}}^{k+1,(i)} + \hat{\T{Q}}^{k,(i)})$ is the SVD of $(\hat{\T{X}}^{k+1,(i)} + \hat{\T{Q}}^{k,(i)})$, then the solution to (\ref{eq:shrink2}) is $U D_\tau (S) V^\text{T}$, where $D_\tau (S) = \text{diag}(S_{i,i} - \tau)_+$ for some positive constant $\tau$ and $``+"$ means keeping the positive part. This is equivalent to multiplying $(1-\frac{\tau}{S_{i,i}})_+$ to the $i_{th}$ singular value of $S$. So each frontal slice of $\hat{\T{Z}}^{k+1}$ can be calculated using this shrinkage on each frontal slice of $(\hat{\T{X}}^{k+1} + \hat{\T{Q}}^{k})$. Let  $\T{U} * \T{S}* \T{V}^\text{T} = (\T{X}^{k+1} + \T{Q}^{k})$ be the t-SVD of $(\T{X}^{k+1} + \T{Q}^{k})$ and $\hat{\T{S}}$ be the Fourier transform of $\T{S}$ along the third dimension. Then each element of the singular tubes of $\hat{\T{Z}}^{k+1}$ is the result of multiplying every entry $\hat{\T{S}}(i,i,j)$ with $(1-\frac{\tau}{\hat{\T{S}}(i,i,j)})_+$ for some $\tau > 0$. Since this process is carried out in the Fourier domain, in the original domain it is equivalent to convolving each tube $\T{S}(i,i,:)$ of $\T{S}$ with a real valued tubal vector $\vec{\tau}_i$ which is the inverse Fourier transform of the vector $[ ( 1-\frac{\tau_{i}(1)}{\hat{\T{S}}(i,i,1)})_+ , ( 1-\frac{\tau_{i}(2)}{\hat{\T{S}}(i,i,2)})_+ , ... , ( 1-\frac{\tau_{i}(n_3)}{\hat{\T{S}}(i,i,n_3)})_+ ]$. This operation can be captured by $\T{S}*\T{T}$, where $\T{T}$ is an f-diagonal tensor with $i_{th}$ diagonal tube to be $\vec{\tau}_i$. Then  $\T{Z}^{k+1} = \T{U}* ( \T{S}*\T{T} )* \T{V}^\text{T}$. In summary, the shrinkage operation in the Fourier domain on the singular values of each of the frontal faces is equivalent to performing a \emph{tubal shrinkage} via convolution in the original domain.


\vspace{-3mm} 
\section{Applications - Video data compression and recovery}
\label{sec:video_app}
\vspace{-2mm} 
\subsection{Compressibility of multi-linear data using t-SVD}
\vspace{-2mm} 
\label{sec:compression}
We outline two methods for compression based on t-SVD and compare them with the traditional truncated SVD based approach in this section. Note that we don't compare with truncated HOSVD or other tensor decompositions as there is no notion of optimality for these decompositions in contrast to truncated t-SVD and truncated SVD. 

The use of SVD in matrix compression has been widely studied in \cite{1093309} and \cite{465543}. 
For a matrix $A \in \Real^{m \times n}$ with its SVD $A = USV^T$, the rank $r$ approximation of $A$ is the matrix $A_r = U_r S_r V^\text{T}_r$, where $S_r$ is a $r \times r$ diagonal matrix with $S_r(i,i) = S(i,i), i = 1,2,...,r$. $U_r$ consists of the first $r$ columns of $U$ and $V^\text{T}_r$ consists of the first $r$ rows of $V^\text{T}_r$. The compression is measured by the ratio of the total number of entries in $A$, which is $mn$, to the total number of entries in $U_r$, $S_r$ and $V^\text{T}_r$, which is equal to $(m+n+1)r$. Extending this approach to a third-order tensor $\T{M}$ of size $n_1 \times n_2 \times n_3$, we vectorize each frontal slice and save it as a column, so we get an $n_1 n_2 \times n_3$ matrix. Then the compression ratio of rank $k_1$ SVD approximation is 
\begin{align}
\text{ratio}_{\text{SVD}} = \frac{n_1 n_2 n_3}{n_1 n_2 k_1 + k_1 + n_3 k_1} = \frac{n_1 n_2 n_3}{k_1(n_1 n_2 +n_3+1)}
\end{align}
where $1 \le k_1 \le \text{min}(n_1n_2,n_3)$. Generally even with small $k_1$, the approximation $M_{k_1}$ gets most of the information of $\T{M}$. 

{\bf Method 1} : Based on t-SVD our first method for compression, which we call t-SVD compression, basically follows the same idea of truncated SVD but in the \emph{Fourier domain}. For an $n_1 \times n_2 \times n_3$ tensor $\T{M}$, we use Algorithm 1 to get $\hat{\T{M}}$, $\hat{\T{U}}$, $\hat{\T{S}}$ and $\hat{\T{V}}^\text{T}$. It is known that $\hat{\T{S}}$ is a f-diagonal tensor with each frontal slice is a diagonal matrix. So the total number of f-diagonal entries of $\hat{\T{S}}$ is $n_0 n_3$ where $n_0 = \text{min} (n_1,n_2)$.  We choose an integer $k_2$, $1 \le k_2 \le n_0 n_3$ and keep the $k_2$ largest f-diagonal entries of $\hat{\T{S}}$ then set the rest to be $0$. If $\hat{\T{S}}(i,i,j)$ is set to be $0$, then let the corresponding columns $\hat{\T{U}}(:,i,j)$ and $\hat{\T{V}}^\text{T}(:,i,j)$ also be $0$. We then call the resulting tensors $\hat{\T{U}}_{k_2}$, $\hat{\T{S}}_{k_2}$ and $\hat{\T{V}}^\text{T}_{k_2}$. So the approximation is $\T{M}_{k_2} = \T{U}_{k_2} * \T{S}_{k_2} * \T{V}^\text{T}_{k_2}$ where $\T{U}_{k_2}$, $\T{S}_{k_2}$ and $\T{V^\text{T}}_{k_2}$ are the inverse Fourier transforms of $\hat{\T{U}}_{k_2}$, $\hat{\T{S}}_{k_2}$ and $\hat{\T{V}}^\text{T}_{k_2}$  along the third dimension.  
The compression ratio rate for this method is
\begin{align}
\text{ratio}_\text{t-SVD} = \frac{n_1 n_2 n_3}{(n_1 \times \frac{k_2}{n_3} + n_2 \times \frac{k_2}{n_3} +\frac{k_2}{n_3}) \times n_3} = \frac{n_1 n_2 n_3}{k_2(n_1+n_2+1)}
\end{align}
where $1 \le k_2 \le n_0 n_3$.

{\bf Method 2}:  Our second method for compressing is called t-SVD-tubal compression and is also similar to truncated SVD but in the \emph{t-product} domain. As in \textbf{Theorem 2.3.1}, we take the first $k_3$ tubes $ ( 1\le k_3 \le n_0) $ in $\T{S}$ and get the approximation $\T{M}_{k_3} = \sum_{i=1}^{k_3} \T{U}(:,i,:) * \T{S}(i,i,:) * \T{V}(:,i,:)^\text{T} $.

Compression ratio rate for the second method is 
\begin{align}
\text{ratio}_\text{t-SVD-tubal} = \frac{n_1 n_2 n_3}{(n_1 k_3 + n_2 k_3 + k_3)n_3} = \frac{n_1 n_2}{k_3(n_1 + n_2 +1)}
\end{align}
where $1 \le k_3 \le n_0$.

%
%
%

\begin{figure}[h]
\centering \makebox[0in]{
    \begin{tabular}{c c}
    \includegraphics[width = .16\textwidth]{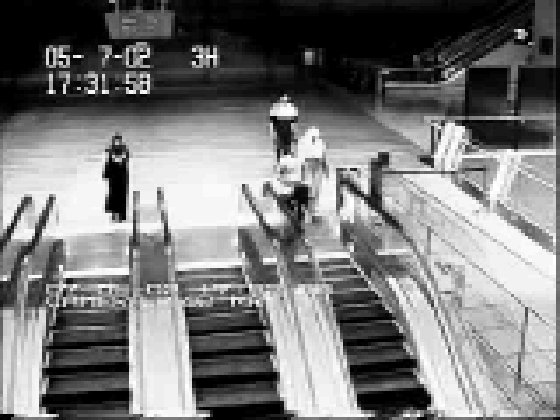}
      \includegraphics[width = .16\textwidth]{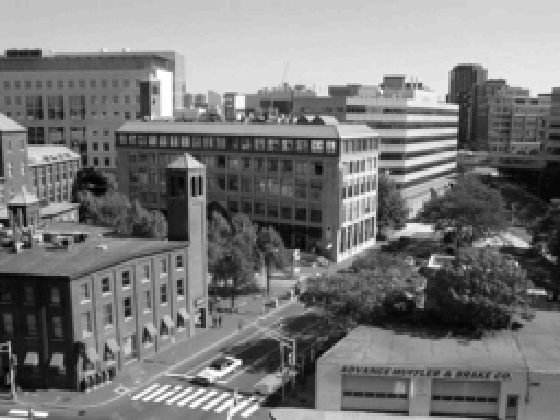}
      \includegraphics[width = .16\textwidth]{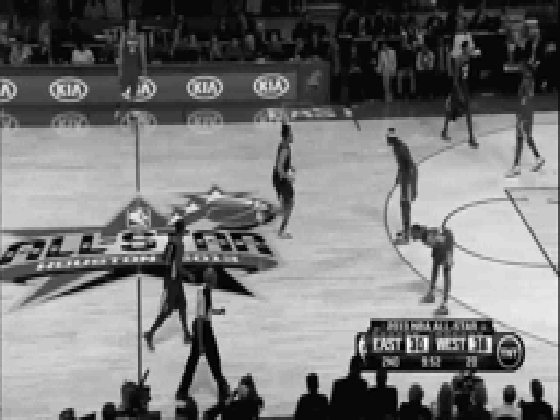}&
      \includegraphics[width = .16\textwidth]{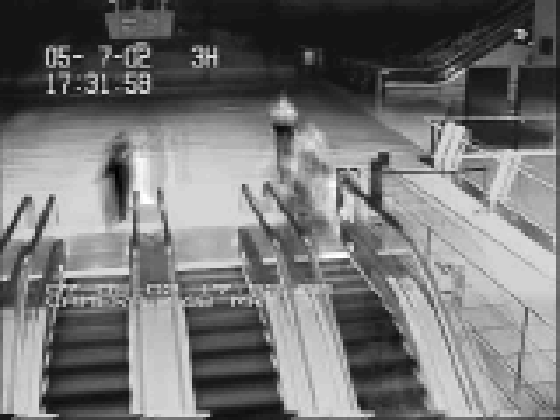}  
      \includegraphics[width = .16\textwidth]{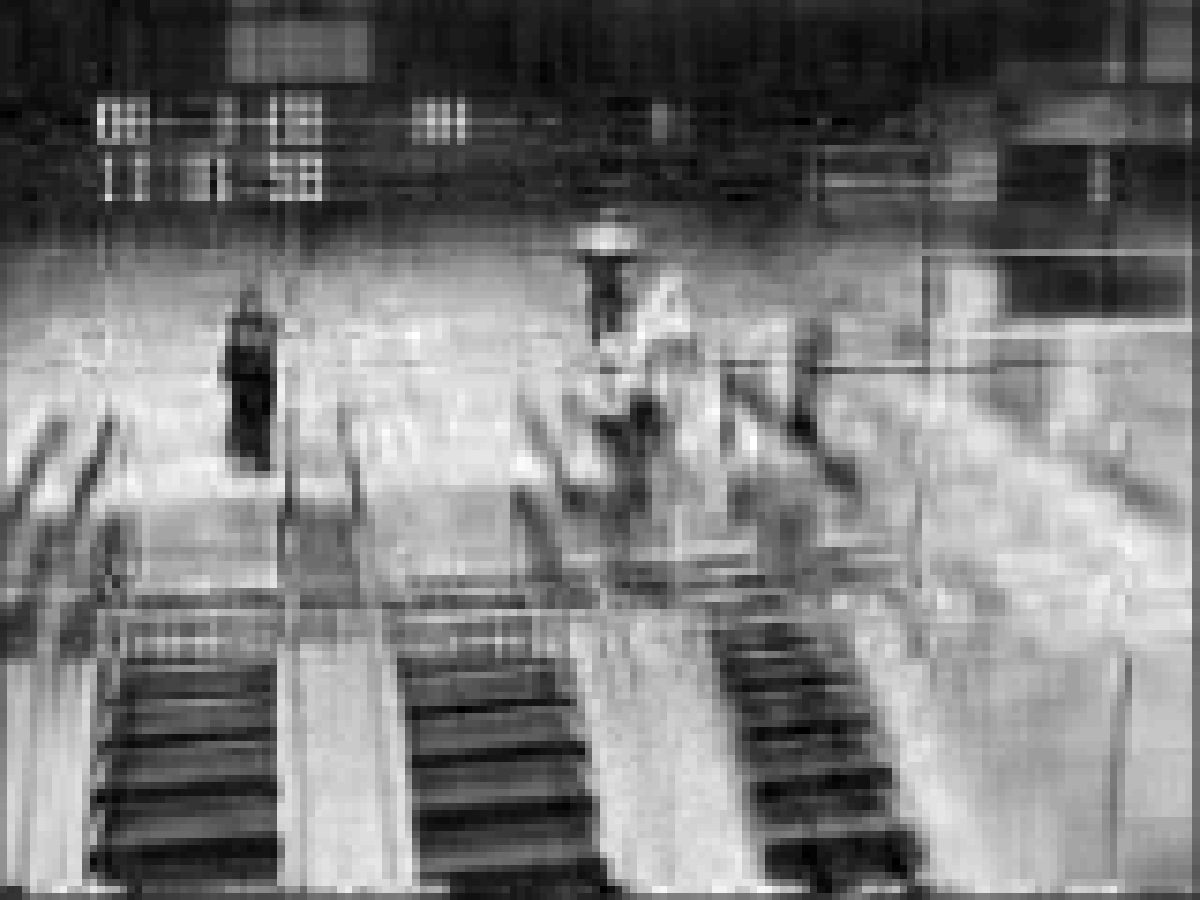} 
      \includegraphics[width = .16\textwidth]{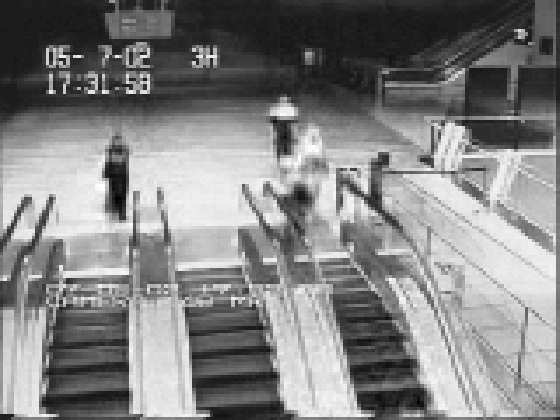} \\
      (a) & (b) \\
      \includegraphics[width = .16\textwidth]{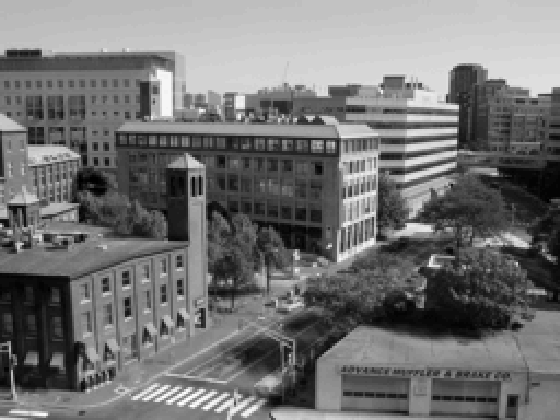}
      \includegraphics[width = .16\textwidth]{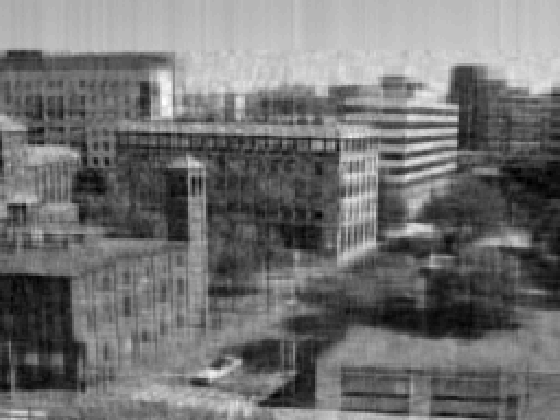}
      \includegraphics[width = .16\textwidth]{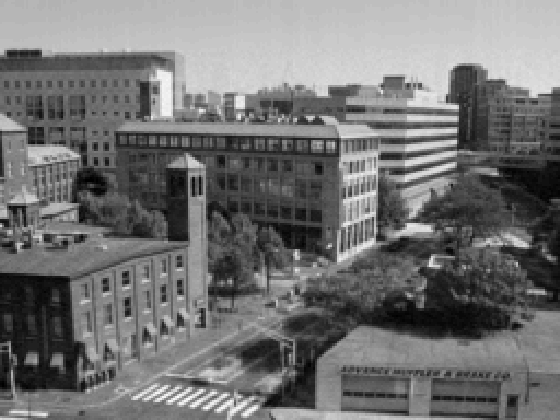} &
      \includegraphics[width = .16\textwidth]{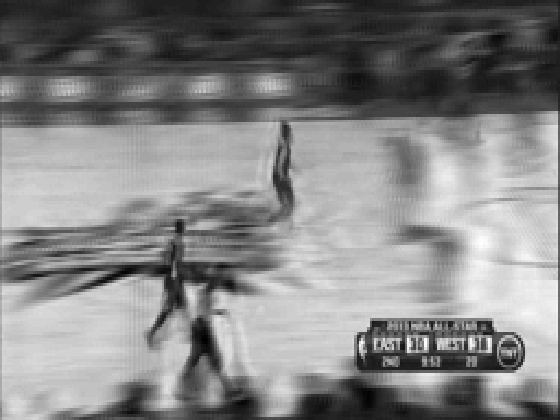}
      \includegraphics[width = .16\textwidth]{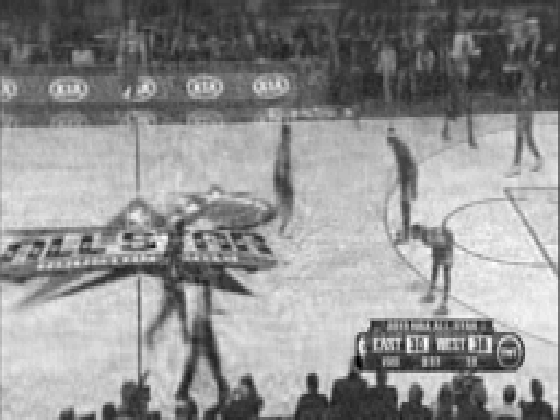}
      \includegraphics[width = .16\textwidth]{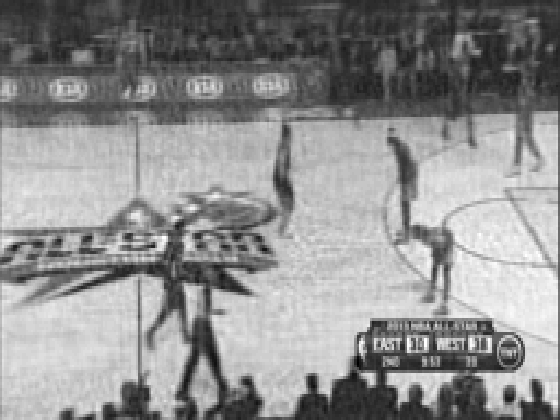} \\
      (c) & (d) 
      \end{tabular}}
  \caption{ (\textbf{a}) Three testing videos: escalator video, MERL video and basketball video. (b) (c) (d) are compression results under compression ratio 5. For (b) (c) (d) from left to right: SVD compression, t-SVD-tubal compression and t-SVD compression}
\label{fig:compress_recovery}
\end{figure}
We now illustrate the performance of SVD based compression, t-SVD compression and t-SVD tubal compression on 3 video datasets shown in Figure~\ref{fig:compress_recovery}-(a).
\vspace{-2mm}
\begin{enumerate}
\itemsep0em
\item The first video, referred to as the Escalator video, (source: http://www.ugcs.caltech.edu/~srbecker/rpca.shtml\#2 ) of size $130 \times 160 \times 50$ (length $\times$ width $\times$ frames) from a stationary camera. 
\item The second video, referred to as the MERL video, is a {\bf \emph{time lapse video}} of size $192 \times 256 \times 38$ also from a stationary camera ({\bf data courtesy}: Dr. Amit Agrawal, Mitsubishi Electric Research Labs (MERL), Cambridge, MA). 
\item The third video, referred to as the Basketball video is a $144 \times 256 \times 80$ video (source: YouTube) with a  {\bf \emph{non-stationary panning camera}} moving from left to right horizontally following the running players. 
\end{enumerate}
\vspace{-2mm}
Figure~\ref{fig:compress_recovery} (b) to (d) show the compression results for the 3 videos when truncated according to vectorized SVD and t-SVD compression (method 1) and t-SVD tensor tubal compression (method 2). In Figure~\ref{fig:compress_rse} we show the relative square error (RSE) comparison for different compression ratio where RSE is defined in dB as $\text{RSE} = 20 \log_{10} ( \|\T{X}_{\text{rec}} - \T{X}\|_\text{F} / \|\T{X}\|_\text{F} )$. In all of the 3 results, the performance of t-SVD compression (method 1) is the best. This implies that tensor multi-rank fits very well for video datasets from both stationary and non-stationary cameras. SVD compression method (based on vectorization) has a better performance over the t-SVD-tubal compression on the Escalator and MERL video. 
However, t-SVD tubal compression (method 2) works much better than SVD compression on the Basketball video. This is because in the videos where the camera is panning or in motion, one frontal slice of the tensor to the next frontal slice can be effectively represented as a shift and scaling operation which in turn is captured by a convolution type operation and t-SVD is based on such an operation along the third dimension.

\begin{figure}[htbp]
\centering \makebox[0in]{
    \begin{tabular}{c c c}
      \includegraphics[width = .33\textwidth]{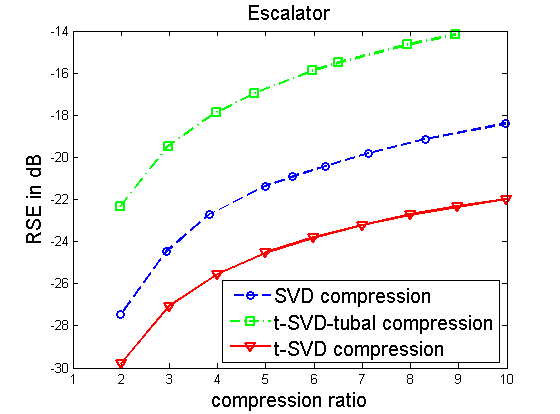}
      
      \includegraphics[width = .33\textwidth]{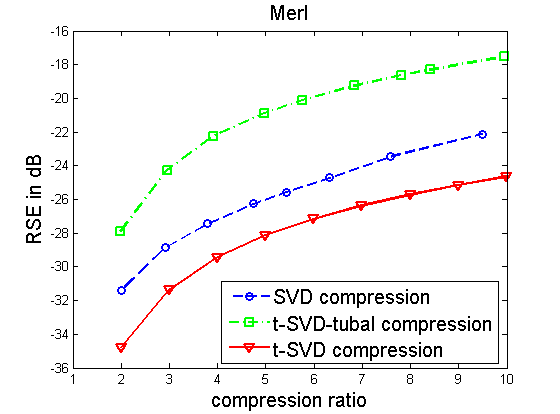}
      
      \includegraphics[width = .33\textwidth]{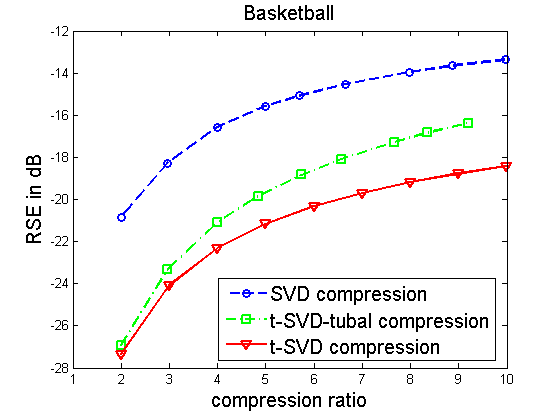}
      
      \end{tabular}}
  \caption{ Compression ratio and RSE comparison for 3 videos.}
\label{fig:compress_rse}
\end{figure}

%
%

\begin{figure}[htbp]
\centering \makebox[0in]{
    \begin{tabular}{c c c c }
      \includegraphics[width = .20\textwidth]{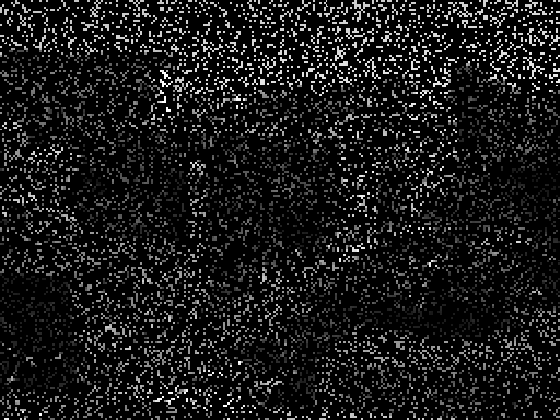}  
      \includegraphics[width = .20\textwidth]{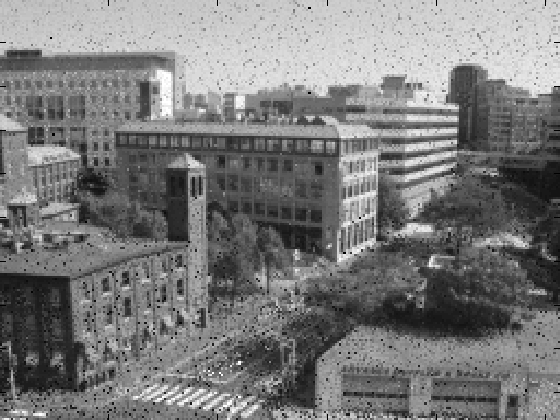}      
      \includegraphics[width = .20\textwidth]{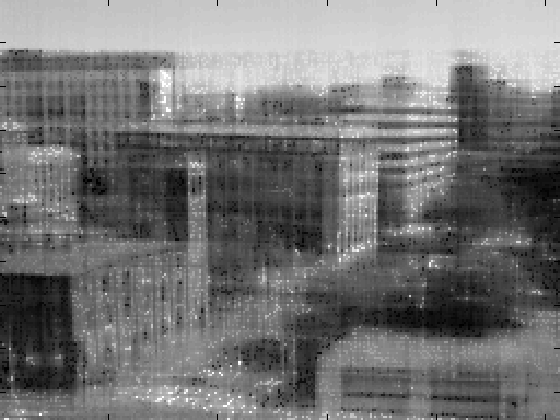}      
      \includegraphics[width = .20\textwidth]{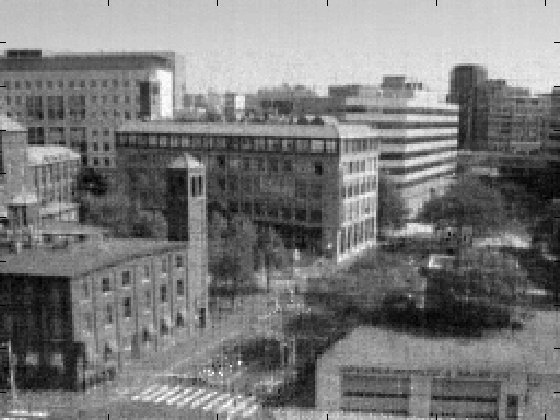} \\     
   
      \includegraphics[width = .20\textwidth]{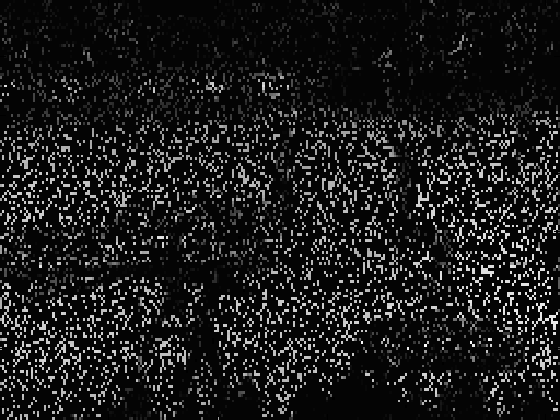}     
      \includegraphics[width = .20\textwidth]{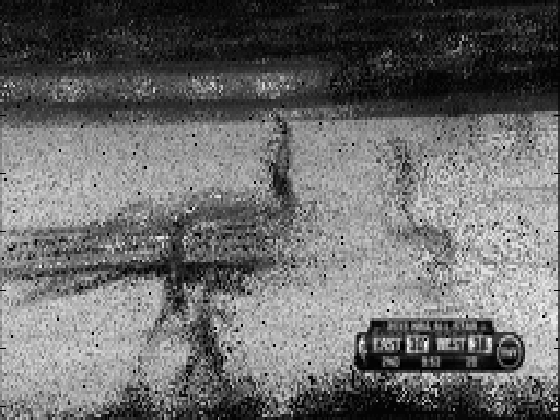}    
      \includegraphics[width = .20\textwidth]{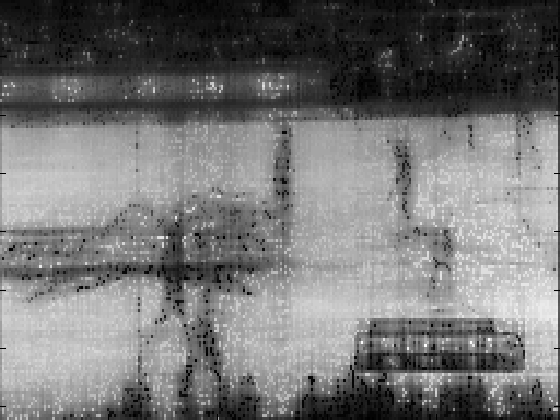}     
      \includegraphics[width = .20\textwidth]{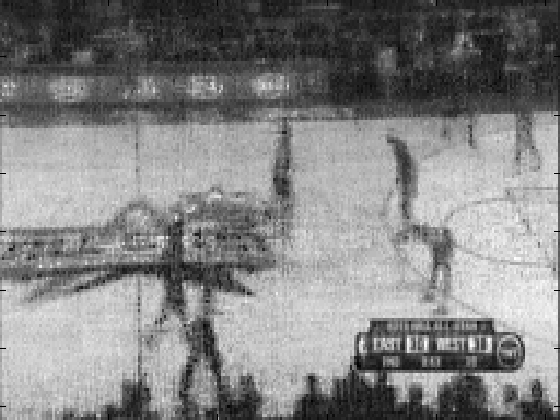} 
      \end{tabular}}
  \caption{ Compression results for MERL and Basketball video. From left to right: Sampled video($20\%$), Nuclear norm minimization (vectorization and SVD based) result, LRTC result, TNN minimization result. \textbf{First row}: MERL video. \textbf{Second row}: Basketball video.}
\label{fig:recovery}
\end{figure}

\vspace{-2mm} 
\subsection{Performance evaluation of video data completion}
\vspace{-2mm} 
In this part we test 3 algorithms for video data completion from randomly missing entries:  TNN minimization of Section \ref{sec:tcompletion}, Low Rank Tensor Completion (LRTC) algorithm in \cite{Ji_PAMI12}, which uses the notion of tensor-n-rank \cite{GandyRY2011}, and the nuclear norm minimization on the vectorized video data using the algorithm in \cite{Cai_SVT}. As an application of the t-SVD to higher order tensor we also show performance on a 4-D color Basketball video data of size $144 \times 256 \times 3 \times 80$. 

Figures ~\ref{fig:recovery} and \ref{fig:basketball_color} show the results of recovery using the 3 algorithms. Figure~\ref{fig:recovery_rse} shows the RSE (dB) plots for sampling rates ranging from $10\%$ to $90\%$ where the sampling rate is defined to be the percentage of pixels which are known. Results from the figures show that the TNN minimization algorithm gives excellent reconstruction over the LRTC and Nuclear norm minimization.  These results are in line with the compressibility results in Section~\ref{sec:compression}.

\begin{figure}[htbp]
\centering \makebox[0in]{
    \begin{tabular}{c c c }
      \includegraphics[width = .22\textwidth]{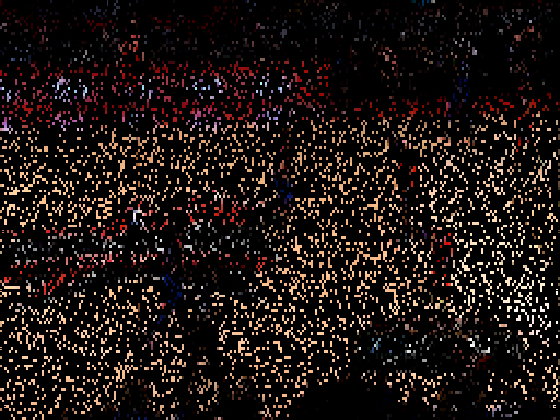}  \hspace{2mm}
      \includegraphics[width = .22\textwidth]{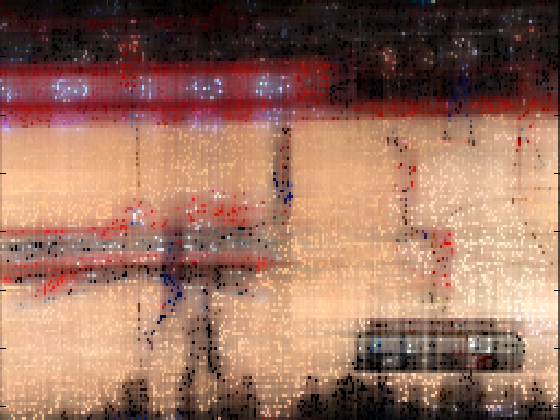}     \hspace{2mm}
      \includegraphics[width = .22\textwidth]{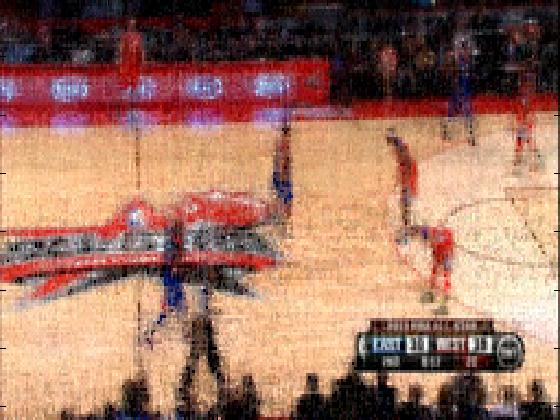} \\
      \end{tabular}}
  \caption{ Recovery for color basketball video: \textbf{Left}: Sampled Video($10\%$). \textbf{Middle}: LRTC recovery. \textbf{Right}: Tensor-nuclear-norm minimization recovery }
\label{fig:basketball_color}
\end{figure}    

\begin{figure}[htbp]
\centering \makebox[0in]{
    \begin{tabular}{c c }
    \includegraphics[width = .45\textwidth]{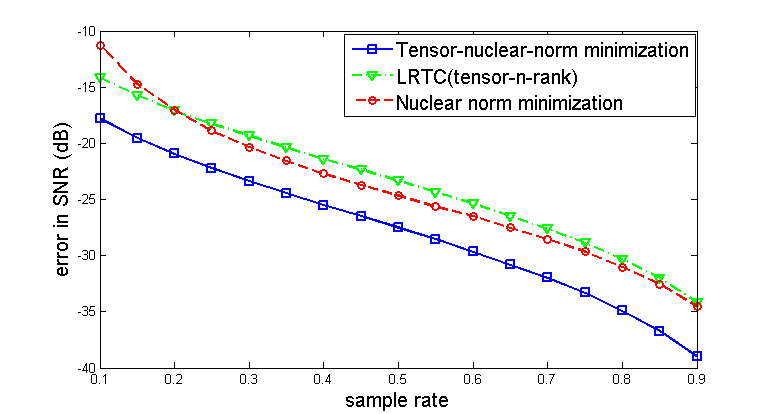} \hspace{5mm}
    \includegraphics[width = .45\textwidth]{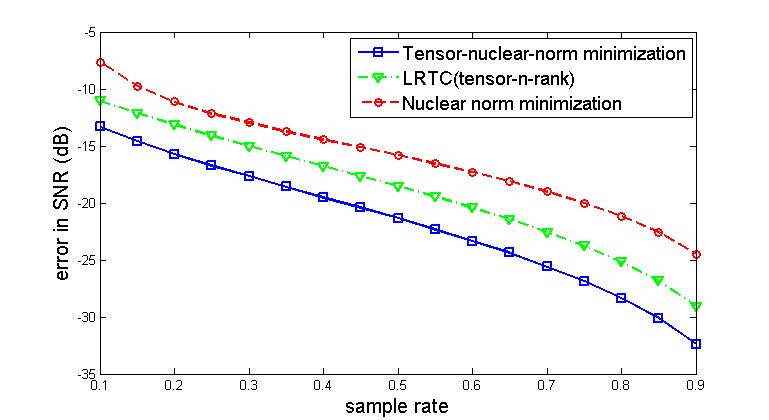}
    \end{tabular}}
  \caption{ RSE (dB) plot against sampling rate \textbf{Left}: MERL video. \textbf{Right}: Basketball video }
\label{fig:recovery_rse}
\end{figure}


\vspace{-4mm} 
\section{Future work - collaborative filtering of multi-linear data}
\label{sec:future_work}
\vspace{-2mm} 
Our work has direct implications to study prediction and recovery of \emph{multi-linear} signals. For example, similar to the scenario of \emph{collaborative filtering} with bounded trace norm for 2-D matrices considered in \cite{Hazan_JMLR12}, one can consider a multi-linear version of the problem and look at collaborative filtering with bounded tensor norms which in turn are derived from the particular tensor decomposition, namely the t-SVD. Due to its optimality properties as well as algebraic grounding, one can look for similar identifiability results similar to the matrix case. Extension of our algorithm in Section \ref{sec:tcompletion} to address this sequential online setting for collaborative filtering of multi-linear data and providing provably optimal learning guarantees is an important research direction. 

%

\newpage 

\pagestyle{empty}
\bibliographystyle{IEEEtran}
\bibliography{bibtensor,SAbib,elmbib,ZZbib}

\end{document}